\documentclass{article}
\usepackage{amsthm, amssymb}
\usepackage{natbib}
\usepackage{color}
\usepackage{dsfont, amsmath, amsmath}
\usepackage{times}
\usepackage{tikz, comment}
\usepackage{pgfplotstable}
\usepackage{enumerate}
\usepackage{enumitem}
\usetikzlibrary{positioning, shapes}

\newtheorem{theorem}{Theorem}
\newtheorem*{theorem*}{Theorem}
\newtheorem{lemma}{Lemma}
\newtheorem*{lemma*}{Lemma}

\newtheorem{innercustomthm}{Theorem}
\newenvironment{repthm}[1]
  {\renewcommand\theinnercustomthm{#1}\innercustomthm}
  {\endinnercustomthm}
\newtheorem{innercustomlemma}{Lemma}
\newenvironment{replem}[1]
  {\renewcommand\theinnercustomlemma{#1}\innercustomlemma}
  {\endinnercustomlemma}

\definecolor{shadecolor}{RGB}{255,255,128}
\definecolor{shadecolor}{gray}{0.92}

\newcommand{\toP}{\stackrel{\scriptscriptstyle P}{\to}}

\DeclareMathOperator*{\plim}{plim}

\newcommand{\eps}{\varepsilon}

\title{Simultaneous Control of All False Discovery Proportions in Large-Scale Multiple Hypothesis Testing}
\author{Jelle J. Goeman$^{1*}$ \and Rosa J. Meijer$^{2}$ \and Thijmen J.P. Krebs$^3$ \and Aldo Solari$^4$}

\begin{document}

\maketitle
\footnotetext[1]{Leiden University Medical Center, Department of Medical Statistics and Bioinformatics, Leiden, The Netherlands}
\footnotetext[2]{Statistics Netherlands, Den Haag, The Netherlands}
\footnotetext[3]{Delft University of Technology, Delft, The Netherlands}
\footnotetext[4]{University of Milano-Bicocca, Department of Economics, Management and Statistics, Milan, Italy}
\renewcommand{\thefootnote}{\fnsymbol{footnote}}
\footnotetext[1]{Corresponding author. Email: j.j.goeman@lumc.nl}

\begin{abstract}
  Closed testing procedures are classically used for familywise error rate (FWER) control, but they can also be used to obtain simultaneous confidence bounds for the false discovery proportion (FDP) in all subsets of the hypotheses. In this paper we investigate the special case of closed testing with Simes local tests. We construct a novel fast and exact shortcut which we use to investigate the power of this method when the number of hypotheses goes to infinity. We show that, if a minimal amount of signal is present, the average power to detect false hypotheses at any desired FDP level does not vanish. Additionally, we show that the confidence bounds for FDP are consistent estimators for the true FDP for every non-vanishing subset. For the case of a finite number of hypotheses, we show connections between Simes-based closed testing and the procedure of Benjamini and Hochberg.
\end{abstract}

\section{Introduction}

Multiple testing methods controlling familywise error (FWER) or false discovery rate (FDR) differ sharply with respect to their power properties for large-scale multiple testing problems. As the number of hypotheses $m$ increases, the power per hypothesis of FWER-controlling procedures such as Bonferroni vanishes. In contrast, if sufficient signal is present, FDR-controlling procedures such as BH \citep{Benjamini1995} reject each false hypothesis with non-vanishing power \citep{Chi2007}.

The method of closed testing \citep{Marcus1976} is a cornerstone of FWER control. Many well-known FWER-controlling methods can be constructed as special cases of this general method. However, \cite{Goeman2011} showed that closed testing procedures can also be used to provide simultaneous upper confidence bounds for the false discovery proportion (FDP) in all $2^m$ subsets of the hypotheses. These additional inferences come for free in addition to the FWER-statements traditionally obtained from closed testing. However, they drastically increase the scope of the closed testing procedure: simultaneous confidence statements for all subsets implies FWER control over exponentially many null hypotheses, each bounding the FDP (i.e.\ the proportion of true null hypotheses) in a subset by a constant.

The question now arises what is the behavior of FDP confidence bounds from a closed testing procedure as the number of hypotheses grows. On the one hand, for a method with FWER control over exponentially many hypotheses it may be expected that interesting statements will vanish relative to the size of the multiple testing problem as with FWER-based methods. On the other hand, we may also expect non-vanishing power, because, like FDR, the FDP criterion scales well with the size of the multiple testing problem.

We investigate this question for the important case of closed testing with Simes local tests. This procedure is valid whenever the Simes inequality holds for the set of all true hypotheses, which happens under independence of $p$-values and certain forms of positive dependence \citep{Sarkar2008, Finner2014}. The same assumption is also necessary for the validity of BH as an FDR-controlling method. Closed testing with Simes local tests is the basis of the well-known FWER-controlling procedures of \cite{Hommel1988} and \cite{Hochberg1988}, procedures that are strictly more powerful than Bonferroni-Holm, but not by a large margin  \citep{Goeman2014}. Since the Simes test is non-consonant, it is known that the FDP upper confidence bounds from this method are sometimes smaller than those arising trivially from the FWER results \citep{Goeman2011, Goeman2014}.

The structure of the paper is as follows. We first derive a novel exact shortcut that allows quick computation of FDP confidence bounds in large problems. This shortcut improves upon earlier shortcuts for this method by \cite{Goeman2011} in speed, power and flexibility. From this shortcut we prove a result on the tail probabilities of FDP of BH-rejected sets, improving upon an earlier result by \cite{Lehmann2005}. This establishes a strong connection between BH and Simes-based closed testing, which suggests that the FDP-bounds might scale like BH in large multiple testing problems. To formalize this, we assume an Efron-style mixture model \citep{Efron2012} for the $p$-values, and show that a similar result to that of \cite{Chi2007} also holds for Simes-based closed testing: provided a minimal amount of signal is present, the average power to detect each false hypothesis does not vanish as $m$ increases. Finally, our main result is that under Efron's model the FDP upper confidence bound is a consistent estimator for the FDP. Crucially, this consistency holds when $m$ goes to infinity at an arbitrary rate with the sample size $n$ as long as the selected set is not too small. All proofs are found in the appendix.

\section{Closed testing and FDP confidence}

We briefly revisit \cite{Goeman2011} to set the stage and introduce notation. Many quantities that will be introduced depend on the significance level $\alpha$. When we introduce quantities we will often subscript them by $\alpha$ to make this explicit, but we will generally suppress such subscripts for brevity when it is unambiguous.

Suppose we have $m$ hypotheses $H_1, \ldots, H_m$ that we are interested in testing. Some of the hypotheses are true, while others are false. Denote by $T \subseteq \{1,\ldots,m\}$ the index set of true hypotheses. The aim of a multiple testing procedure is to come up with an index set $S \subseteq \{1,\ldots,m\}$ of rejected or `selected' hypotheses, the \emph{discoveries}. This selection should avoid \emph{false discoveries}, which are hypotheses with indices in $T \cap S$. We denote the number of false discoveries for a selected set $S$ as
\[
\tau(S) = |T \cap S|,
\]
where $|\cdot|$ denotes the size of a set.

\cite{Goeman2011} show how to calculate upper confidence bounds $t_\alpha(S)$ for $\tau(S)$ from the closed testing procedure of \cite{Marcus1976}. In the general closed testing procedure, the collection of hypotheses is augmented with all possible intersection hypotheses $H_I=\bigcap_{i\in I} H_i$, with $I\subseteq\{1,\ldots,m\}$. An intersection hypothesis $H_I$ is true if and only if $H_i$ is true for all $i\in I$. Note that $H_i = H_{\{i\}}$, so that all original hypotheses, known as \emph{elementary hypotheses}, are also intersection hypotheses.

The closed testing procedure starts by testing all intersection hypotheses with a \emph{local test}. An intersection hypothesis $H_I$ is subsequently rejected by the closed testing procedure if and only if all intersection hypotheses $H_J$ with $I\subseteq J$ are rejected by the local test. It will be helpful to use the notation $\mathcal{I} = 2^{\{1,\ldots,m\}}$ for the collection of all subsets of $\{1,\ldots,m\}$ and $\mathcal{T} = \{I \in \mathcal{I}\colon I \subseteq T\}$ for the collection of index sets corresponding to true intersection hypotheses. We define $\mathcal{U}_\alpha$ as the collection of all $I \in \mathcal{I}$ such that $H_I$ is rejected by the local test (defined in the next section). For each $I$, $H_I$ is rejected by the closed testing procedure if and only if $I \in \mathcal{X}_\alpha$, where
\[
\mathcal{X}_\alpha = \{I \in \mathcal{I}\colon \textrm{$J \in \mathcal{U}_\alpha$ for all $J \supseteq I$}\}.
\]

We will assume that the local test rejects $H_T$, the intersection of all true hypotheses, with probability at most $\alpha$:
\begin{equation} \label{eq_ass_U}
\mathrm{P}(T \notin \mathcal{U}_\alpha) \geq 1-\alpha.
\end{equation}
Under this assumption the closed testing procedure controls the familywise error rate for all intersection hypotheses $H_I$, that is $\mathrm{P}(\mathcal{T} \cap \mathcal{X}_\alpha = \emptyset) \geq 1-\alpha$. This follows because $T \notin \mathcal{U}_\alpha$ implies that $I \notin \mathcal{X}_\alpha$ for every $I \subseteq T$, which in turn implies $\mathcal{T} \cap \mathcal{X}_\alpha = \emptyset$.

For the confidence bound given by
\begin{equation}\label{t_alpha}
t_\alpha(S) = \max \{|I|\colon I\subseteq S, I \notin \mathcal{X}_\alpha\},
\end{equation}
\cite{Goeman2011} prove that
\begin{equation} \label{eq_confidence}
\mathrm{P}(\textrm{$\tau(S) \leq t_\alpha(S)$ for all $S \in \mathcal{I}$}) \geq 1-\alpha.
\end{equation}
This result holds for any closed testing procedure, i.e.\ for any $\mathcal{U}$ with the property (\ref{eq_ass_U}). The bounds can also be formulated in terms of the false discovery proportion (FDP), given by $\pi(S) = \tau(S)/|S|$ if $S\neq\emptyset$, and 0 otherwise. We find
\[
q_\alpha(S) = \frac{t_\alpha(S)}{|S|},
\]
also defined as 0 if $S=\emptyset$, as an immediate simultaneous upper $(1-\alpha)$-confidence bound for $\pi(S)$. We denote $d_\alpha(S) = |S|-t_\alpha(S)$ as the simultaneous lower $(1-\alpha)$-confidence bound for the number of true discoveries $|S|-\tau(S)$.

As argued by \cite{Goeman2011}, a data analyst can study the data, choose any set $S$ of interest and use (\ref{t_alpha}) to obtain a confidence bound for the number of false discoveries in that set. The choice of $S$ can be revised as many times as desired, since the simultaneity of the confidence bounds over $S$ guarantees their post hoc validity. Moreover, the simultaneity over $S$ also guarantees FWER control if multiple sets are reported.

\section{Simes and non-consonant rejections} \label{sec_Non-consonant}

In the remainder of this paper we study $t_\alpha(S)$ and its properties for the special case that the local test is a Simes test. The Simes test is based on the $p$-values of elementary hypotheses. We suppose that we have raw (unadjusted) $p$-values $p_1,\ldots,p_m$ for the hypotheses $H_1,\ldots H_m$. The Simes local test rejects an intersection hypothesis $H_I$ (i.e.\ $I \in \mathcal{U}_\alpha$) if and only if there is at least one $1 \leq i \leq |I|$ for which
\begin{equation} \label{eq_Simestest}
|I|p_{(i\mathbin{:}I)} \leq i\alpha,
\end{equation}
where for any $I \in \mathcal{I}$ and $1 \leq i \leq |I|$, we define $p_{(i\mathbin{:}I)}$ as the $i$th smallest $p$-value among the multiset $\{p_i\colon i\in I\}$. Equivalently, for $I \neq \emptyset$, Simes rejects $H_I$ whenever
\[
\min_{1\leq i \leq |I|}p_{(i\mathbin{:}I)}/i \leq \alpha/|I|.
\]

In the case of the Simes local test assumption (\ref{eq_ass_U}) is better known as the \emph{Simes inequality}:
\begin{equation}\label{Simesineq}
\mathrm{P}\big(\textrm{$|T| p_{(i\mathbin{:}T)} \leq i\alpha$ for at least one $i\in\{1,\ldots,|T|\}$}\big)\leq \alpha.
\end{equation}
The Simes inequality holds for independent $p$-values and under certain forms of positive correlation. There is extensive and growing literature on sufficient conditions for the Simes inequality, which we will not revisit here \citep{Benjamini2001, Rodland2006, Sarkar2008, Finner2014, Bodnar2017}. Assumption (\ref{Simesineq}) is also necessary for the validity of FDR control for the procedure of \cite{Benjamini1995}.

The literature on closed testing makes a distinction between \emph{consonant} and \emph{non-consonant} procedures \citep{Henning2015, Brannath2010, Bittman2009, Gou2014}. A consonant closed testing procedure has the property that for every $I$ rejection of $H_I$ by the closed testing procedure implies rejection of $H_i$ for at least one $i \in I$. Procedures without this property are non-consonant. From the perspective of FWER control, consonance is a desirable property. However, from the perspective of FDP control this is not necessarily the case.

Non-consonant procedures can have FDP confidence bounds that are not trivially implied by FWER control. Let
\begin{equation} \label{def_R}
R_\alpha = \{1\leq i \leq m\colon \{i\} \in \mathcal{X}_\alpha\},
\end{equation}
be the index set of elementary hypotheses rejected by the closed testing procedure. For the case of Simes local tests $R$ corresponds to the rejected set of \cite{Hommel1988}. It is easy to see that in general
\begin{equation} \label{eq_SR}
t_\alpha(S) \leq |S \setminus R|,
\end{equation}
the trivial bound on the number of false discoveries implied by the FWER control property of $R$. We have $t_\alpha(S) < |S \setminus R|$ for at least one $S$ if and only if at least one non-consonant rejection has been made, where a non-consonant rejection is a rejection of $H_I$ for which no $H_i$, $i\in I$ are rejected.

If $t_\alpha(S) < |S \setminus R|$ this implies that there are some false null hypotheses in $S$ that we can detect, but not pinpoint at the same confidence level: we know they are present in $S$, but we cannot say where in $S$ they are. This relates closely to the difference between estimable and detectable effects, as discussed by \cite{Donoho2004}.

The Simes test can make non-consonant rejections. A simple example, illustrated in Figure \ref{consonant}, serves to demonstrate this. Take for example $m=4$, with $\alpha/2 < p_1 \leq p_2 \leq p_3 \leq 2\alpha/3$, and $p_4 > \alpha$. Then it is easy to verify that the Simes test rejects $H_I$ for all $I$ except $I=\{1,4\}$, $\{2,4\}$, $\{3,4\}$ and $\{4\}$. Consequently, the closed testing procedure rejects $H_I$ with $I=\{1,2\}$, $\{1,3\}$ and $\{2,3\}$, and all $H_I$ with $|I|>2$, but fails to reject $H_1$, $H_2$, $H_3$ as well as those intersection hypotheses for which the Simes test failed. The end result is clearly non-consonant. We have $R=\emptyset$, and for e.g.\ $S = \{1,2,3\}$ we have $t_\alpha(S) = 1$ while $|S \setminus R|=3$.

The example of Figure~\ref{consonant} may seem like a contrived example. Indeed, non-consonant rejections in Hommel's procedure are rare if $m$ is small. They occur very frequently if $m$ is large, however, as we shall see later in this paper.

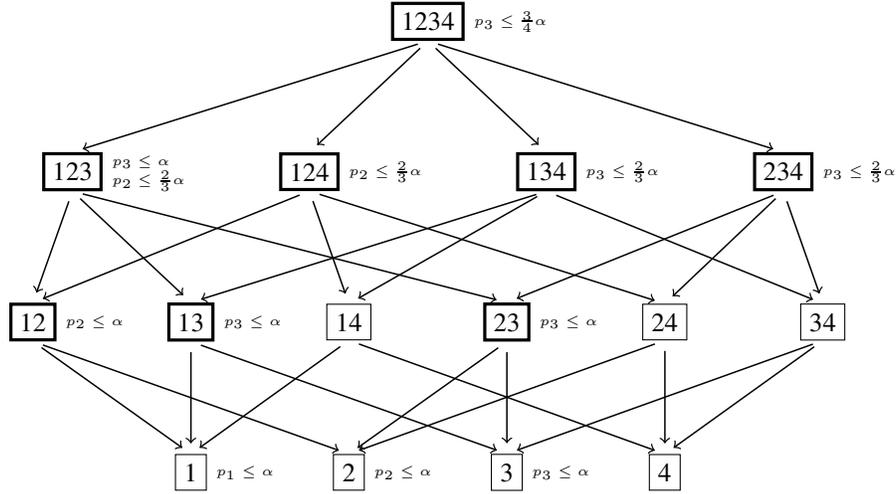
\begin{figure}[!ht]
\centering
\begin{tikzpicture}[xscale=1.05]
  \path (7,8) node[very thick, draw, label={right:\tiny $p_3 \leq \frac34\alpha$}] (1234) {1234};
  \path (2.5,6) node[draw, very thick, label={right:\tiny $\!\!\!\!\begin{array}{l} p_3 \leq \alpha\\p_2 \leq \frac23\alpha \end{array}$}] (123) {123};
  \path (5.5,6) node[draw, very thick, label={right:\tiny $p_2 \leq \frac23\alpha$}] (124) {124};
  \path (8.5,6) node[draw, very thick, label={right:\tiny $p_3 \leq \frac23\alpha$}] (134) {134};
  \path (11.5,6) node[draw, very thick, label={right:\tiny $p_3 \leq \frac23\alpha$}] (234) {234};
  \path (2,4) node[draw, very thick, label={right:\tiny $p_2 \leq \alpha$}] (12) {12};
  \path (4,4) node[draw, very thick, label={right:\tiny $p_3 \leq \alpha$}] (13) {13};
  \path (6,4) node[draw] (14) {14};
  \path (8,4) node[draw, very thick, label={right:\tiny $p_3 \leq \alpha$}] (23) {23};
  \path (10,4) node[draw] (24) {24};
  \path (12,4) node[draw] (34) {34};
  \path (4,2) node[draw, label={right:\tiny $p_1 \leq \alpha$}] (1) {1};
  \path (6,2) node[draw, label={right:\tiny $p_2 \leq \alpha$}] (2) {2};
  \path (8,2) node[draw, label={right:\tiny $p_3 \leq \alpha$}] (3) {3};
  \path (10,2) node[draw] (4) {4};

  \begin{scope}[shorten >= 4pt,shorten <= 4pt, latex-, ->, line width=0.2mm]
  \draw  (1234.south) -- (123.north);
  \draw  (1234.south) -- (124.north);
  \draw  (1234.south) -- (134.north);
  \draw  (1234.south) -- (234.north);
  \draw  (123.south) -- (12.north);
  \draw  (123.south) -- (13.north);
  \draw  (123.south) -- (23.north);
  \draw  (124.south) -- (12.north);
  \draw  (124.south) -- (14.north);
  \draw  (124.south) -- (24.north);
  \draw  (134.south) -- (13.north);
  \draw  (134.south) -- (14.north);
  \draw  (134.south) -- (34.north);
  \draw  (234.south) -- (23.north);
  \draw  (234.south) -- (24.north);
  \draw  (234.south) -- (34.north);
  \draw  (12.south) -- (1.north);
  \draw  (12.south) -- (2.north);
  \draw  (13.south) -- (1.north);
  \draw  (13.south) -- (3.north);
  \draw  (14.south) -- (1.north);
  \draw  (14.south) -- (4.north);
  \draw  (23.south) -- (2.north);
  \draw  (23.south) -- (3.north);
  \draw  (24.south) -- (2.north);
  \draw  (24.south) -- (4.north);
  \draw  (34.south) -- (3.north);
  \draw  (34.south) -- (4.north);
  \end{scope}
\end{tikzpicture}
\caption{Example of non-consonant rejections in Hommel's procedure based on $\alpha/2 < p_1  \leq p_2 \leq p_3 \leq 2\alpha/3$, and $p_4 > \alpha$. Nodes in the graph represent intersection hypotheses, labeled by their corresponding index set. Edges represent subset relationships. Hypotheses rejected by the Simes test have the reason(s) for their rejection given next to them. Hypotheses rejected by the closed testing procedure are marked in bold. A non-consonant rejection is e.g.\ $H_{\{1,3\}}$, since it is rejected but $H_1$ and $H_3$ are not.} \label{consonant}
\end{figure}

\section{A shortcut} \label{sec:shortcut}

Naive application of a closed testing procedure requires $2^m$ local tests to be performed, which severely limits the usefulness of the procedure in large problems. For this reason \emph{shortcuts} have been developed, which are algorithms that limit the computation time of the closed testing procedure. Shortcuts come in two flavors: exact shortcuts, which lead to the same rejections as the closed testing procedure, and approximate or conservative shortcuts, which lead to at most the same rejections as the closed testing procedure, but possibly fewer. Shortcuts always apply to specific local tests; there are no general shortcuts for the closed testing procedure.

\cite{Hommel1988} provided the first shortcut for closed testing with Simes local tests, but that shortcut only applies to elementary hypotheses. \cite{Goeman2011, Goeman2013} presented an approximate shortcut for FDP confidence bounds for closed testing with Simes local tests. This shortcut only applies to sets $S$ consisting of the $j$ smallest $p$-values for some $j$. We abbreviate $p_{(i)} = p_{(i\mathbin{:}\{1,\ldots,m\})}$ as usual. Let $r_1, \ldots, r_m$ be a permutation of $1,\ldots, m$ such that $p_{r_i} = p_{(i)}$ for all $1 \leq i \leq m$. We define for $0 \leq i \leq m$ the set corresponding to the $i$ smallest $p$-values as
\[
L_i = \{r_1, r_2, \ldots, r_i\}.
\]
\cite{Goeman2011} showed that $t_\alpha(L_j) \leq |S|-d$ if there is a $d \leq i \leq j$ such that $(m-d+1)p_{(i)} \leq (i-d+1)\alpha$. Applying the shortcut to the example in Figure \ref{consonant}, we find for $L_3=\{1,2,3\}$ that it gives $t_\alpha(L_3) \leq 1$ since $3p_3 \leq 2\alpha$. However, the shortcut fails for $L_2=\{1,2\}$, demonstrating that the shortcut is only approximate. Computation time for this shortcut is quadratic in $j$.

In this section we provide a novel shortcut that improves upon the shortcut of \cite{Goeman2011} in three ways: the novel shortcut is (1.) applicable for all $S$; (2.) exact; and (3.) faster to compute.

Analogous to $L_i$ we define $K_i = \{r_{m-i+1}, \ldots, r_m\}$ for $i=0,\ldots,m$. This is the set of the $i$ largest $p$-values, the `worst case' index set of size $i$. From this we define
\begin{equation} \label{def_h}
h_\alpha = \max\{0 \leq i \leq m\colon K_i \notin \mathcal{U}_\alpha\}.
\end{equation}

The quantity $h$ was first defined by \cite{Hommel1988} and plays a large role in his FWER-controlling procedure. The most important property of $h$ is formalized in the following lemma, already known to Hommel, which says that $h$ is the largest size of an intersection hypothesis that is not rejected by the closed testing procedure.
\begin{lemma} \label{h}
$|I| > h$ implies $I \in \mathcal{X}$.
\end{lemma}
The proof of this lemma is trivial. Combining Lemma \ref{h} with (\ref{t_alpha}) it immediately follows that $t_\alpha(\{1,\ldots,m\}) = h$. Consequently
\begin{equation} \label{hatpi}
\hat\pi_\alpha = h_\alpha/m
\end{equation}
is an $(1-\alpha)$-upper confidence bound for the proportion of true null hypotheses in the multiple testing problem, a quantity for which many conservative estimates are available, but few true confidence bounds \citep[e.g.][]{Langaas2005, Nettleton2006}.

The value of $h$ can also be used to quickly determine whether a hypothesis $H_I$ can be rejected within the closed testing procedure, as shown by the following lemma. This lemma was motivated by a flow chart of Hommel (\citeyear[Fig.\ 3]{hommel1986}). It is a generalization of the procedure of \cite{Hommel1988} to intersection hypotheses. Hommel's procedure is recovered if we apply Lemma \ref{CTeasy} to singleton sets.

\begin{lemma} \label{CTeasy}
$I \in \mathcal{X}$ if and only if there exists some $1 \leq i \leq |I|$ such that
$h p_{(i\mathbin{:}I)} \leq i\alpha$.
\end{lemma}

Note the similarity between the rule of Lemma \ref{CTeasy} and the Simes test (\ref{eq_Simestest}). Effectively the critical values of the local test are multiplied by a factor $h/|I|$ (for $I \neq \emptyset$), emphasizing that the multiple testing burden decreases with $|I|$. Thus, Lemma \ref{CTeasy} suggests that closed testing with Simes local tests is relatively powerful for intersections of many hypotheses, but may be not so powerful for intersections of few hypotheses, and in particular for elementary hypotheses. This is ironic, as it has so far almost exclusively been used for inference on elementary hypotheses.

Lemma \ref{CTeasy} is instrumental for deriving the actual shortcut for $t_\alpha(S)$, defined in (\ref{t_alpha}), given in the following theorem.

\begin{theorem} \label{shortcut}
If $S \neq\emptyset$ we have
\begin{equation} \label{eq_tS_short}
d_\alpha(S) = \max_{1\leq u \leq |S|} 1-u+|\{i\in S\colon hp_i \leq u\alpha\}|.
\end{equation}
\end{theorem}

This shortcut is exact and applicable for all $S$. As an illustration, we apply it to the example of Figure \ref{consonant}, where $h=2$. For $S=\{1,3\}$ the maximal value is obtained for $u=2$, where we find $t_\alpha(S) = 1$.

Computation of $d_\alpha(S)$ using (\ref{eq_tS_short}) requires first computing $h$. It can be shown that $K_i \in \mathcal{U}$ implies $K_j\in \mathcal{U}$ if $j \geq i$, and this fact can be exploited to calculate $h$ for a fixed $\alpha$ using bisection in $O(m\log m)$ time. In fact, \cite{Meijer2017} even provide an algorithm that can be used to calculate $h$ for all $\alpha$ simultaneously in $O(m\log m)$ time. This allows computation of $d_\alpha(S)$ in $O(m\log m)$ time. If we need to calculate $d_\alpha(S)$ for many $S$, we note that $h$ does not depend on $S$, so this quantity needs to be calculated only once. Computation time for $S$ is linear in $|S|$ if $h$ has already been calculated, as follows. The $t_\alpha(S)$ depends on $p_i$ only through $c_i = \lceil hp_i/\alpha\rceil$, for all $i \in S$ (if $\alpha=0$ we take $c_i=0$ if $hp_i=0$ and $c_i=\infty$ otherwise). The $c_i$, $i\in S$, can be sorted in linear time in $|S|$ by discarding all $c_i > |S|$ and using a counting sort. Evaluating (\ref{eq_tS_short}) on these sorted $c_i$ takes linear time in $|S|$.

Further practical improvements are possible, although these do not change the formal computational complexity. If $|S|>0$, the maximum in (\ref{eq_tS_short}) is attained at most at the value of $u$ where the sum is equal to $|S|$, which happens when $u\alpha \geq h\max_{i\in S} p_i$. Moreover, we may restrict the values of $u$ that need to be considered further to $u \leq z_\alpha$, defined in the following lemma. Note that $z$ does not depend on $S$, and can be calculated in $O(m)$ time after sorting the $p$-values, so the use of Lemma \ref{lemma_z} does not increase the computational complexity.

\begin{lemma} \label{lemma_z}
The maximum of $1 - u + |\{i\in S\colon hp_i \leq u\alpha\}|$ over $1 \leq u \leq |S|$
is attained for some $1\leq u \leq z-m+h+1$, where
\[
z_\alpha = \left\{ \begin{array}{ll}
0, & \textrm{if $h=m$;} \\
\min\big\{ m-h \leq i \leq m \colon hp_{(i)} \leq (i-m+h+1)\alpha \big\}, & \textrm{otherwise.}
\end{array} \right.
\]
\end{lemma}


Clearly, the shortcut of Theorem \ref{shortcut} improves upon the earlier shortcut of \cite{Goeman2011} in computational speed, accuracy and applicability. In the following sections, we will use this shortcut to derive properties of the FDP bounds $q_\alpha(S)$.

\section{Relationship to Benjamini-Hochberg}

The Simes test and the FDR-controlling procedure of \cite{Benjamini1995}, abbreviated BH, share the same critical values. Close relationships between BH and closed testing based on Simes local tests can therefore be expected. A trivial but important relationship, already remarked by \cite{Goeman2011}, is that the event that the two methods make at least one non-trivial statement is identical, and that consequently the methods have the same weak FWER control. In this section we explore two further connections. We define the rejection set of BH at level $\alpha$ by $B_\alpha = L_{b_\alpha}$, with
\begin{equation} \label{def_b}
b_\alpha = \max\{1\leq i\leq m\colon mp_{(i)} \leq i\alpha\},
\end{equation}
taking $b_\alpha=0$ ($B_\alpha=\emptyset$) if the maximum does not exist.

The first connection we make between BH and closed testing with Simes local tests says that all non-trivial bounds $t_\alpha(S)$ relate to hypotheses rejected by BH. We will construct this as a consequence of Lemma \ref{lemma_z}. Let $Z_\alpha = L_z$, where $z$ is defined in Lemma \ref{lemma_z}. The following lemma is a corollary to Lemma \ref{lemma_z}.

\begin{lemma} \label{lemma_Z}
We have
\begin{enumerate}
\item $d_\alpha(S)=d_\alpha(S \cap Z)$ for all $S$;
\item $Z$ is minimal, i.e.\ if $d_\alpha(S)=d_\alpha(S \cap Y)$ for all $S$, then $Z \subseteq Y$;
\item $d_\alpha(Z)=m-h$;
\item $Z\subseteq B$.
\end{enumerate}
\end{lemma}

Lemma \ref{lemma_Z} says that all interesting statements from the closed testing procedure are about hypotheses with indices in $Z$, i.e.\ with $p$-values at most $p_{(z)}$. We call $Z$ the \emph{concentration set} because all interesting results are `concentrated' in $Z$: there is no confidence for any discoveries relating to hypotheses with indices outside $Z$. For a user looking where the interesting findings of the experiment are, it makes sense to restrict to hypotheses with index in $Z$. If the user restricts to a smaller set, he or she may lose out on some discoveries. Since $Z \subseteq B$, all interesting findings of the closed testing procedure at $(1-\alpha)$-confidence are related to BH-rejected hypotheses. A helpful graphical construction of the relationship between $z$ and $b$ is given in Figure~\ref{fig_hzb}.

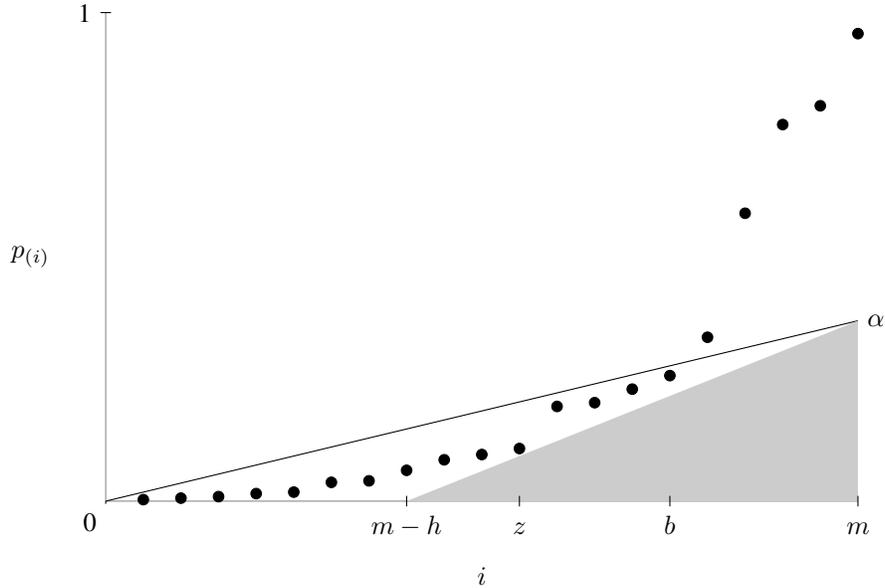
\begin{figure}[!ht]
\centering
\begin{tikzpicture}
    \draw[help lines] (0,0)--(10,0);
    \draw[help lines] (0,0)--(0,6.5) ;
    \draw (-2pt, 6.5) node[left] {1} -- (2pt, 6.5);
    \draw[fill] (0.5,0.02) circle(2pt);
    \draw[fill] (1,0.04) circle(2pt);
    \draw[fill] (1.5,0.06) circle(2pt);
    \draw[fill] (2,0.10) circle(2pt);
    \draw[fill] (2.5,0.12) circle(2pt);
    \draw[fill] (3,0.25) circle(2pt);
    \draw[fill] (3.5,0.27) circle(2pt);
    \draw[fill] (4,0.41) circle(2pt);
    \draw[fill] (4.5,0.55) circle(2pt);
    \draw[fill] (5,0.62) circle(2pt);
    \draw[fill] (5.5,0.70) circle(2pt);
    \draw[fill] (6,1.26) circle(2pt);
    \draw[fill] (6.5,1.31) circle(2pt);
    \draw[fill] (7,1.49) circle(2pt);
    \draw[fill] (7.5,1.67) circle(2pt);
    \draw[fill] (8,2.18) circle(2pt);
    \draw[fill] (8.5,3.83) circle(2pt);
    \draw[fill] (9,5.01) circle(2pt);
    \draw[fill] (9.5,5.26) circle(2pt);
    \draw[fill] (10,6.22) circle(2pt);
    \draw (0, 0) -- (10, 2.4) node[right] (alpha) {$\alpha$};
    \draw[draw=none, fill=gray!40] (4,0) -- (alpha.west) -- (10,0);
    \draw (4, -2pt) node[below] {$m-h$} -- (4, 2pt);
    \draw (5.5, -2pt) node[below] {\phantom{$h$}$z$\phantom{$h$}} -- (5.5, 2pt);
    \draw (7.5, -2pt) node[below] {\phantom{$h$}$b$\phantom{$h$}} -- (7.5, 2pt);
    \draw (10, -2pt) node[below] {\phantom{$h$}$m$\phantom{$h$}} -- (10, 2pt);
    \draw (0,-1pt) node[below left] {0} -- (0,0);
    \path (5,-1) node{$i$};
    \path (-1,3.25) node{$p_{(i)}$};
\end{tikzpicture}
\caption{Graphical illustration of $h$, $z$ and $b$. The points are the $p$-values plotted against their rank. The value of $h$ is the largest such that the gray triangle does not contain any of the points. The value of $z$ is the smallest index $i$ such that $(i+1, p_i)$ is inside the triangle. The value of $b$ is the index of the largest point below the line. The relationship $m-h \leq z \leq b$ is clear.} \label{fig_hzb}
\end{figure}

Combining Lemma \ref{lemma_Z} with (\ref{eq_SR}) it follows that
\[
|S \cap R| \leq d_\alpha(S) \leq |S \cap Z| \leq |S \cap B|.
\]
The confidence bound on the number of discoveries in $S$ is between the number of Hommel-FWER discoveries in $S$ and the number of BH-FDR discoveries in $S$. BH, therefore, bounds the number of discoveries of Simes-based closed testing from above.

On the other hand, from Lemma \ref{lemma_Z} we have $d_\alpha(B) = m-h$. Therefore, if $h<m$ (equivalently $b>0$), we have
\[
q_\alpha(B)\ =\ \frac{b-m+h}{b}\ <\ 1,
\]
which always gives a non-trivial confidence bound for the FDP realized by BH. This makes the trivial statement that $q_\alpha(\{1,\ldots,m\}) < 1$ if $b>0$, made at the start of this section, more precise.

As a second connection between BH and closed testing with Simes, we explore $q_\alpha(B_{q})$ for general $q$ and $\alpha$. This relationship we formulate in Lemma \ref{lemma_BH}. This lemma can be seen as a uniform improvement of the shortcut of \cite{Goeman2011} given in Section \ref{sec:shortcut}, but not as general and exact as that of Theorem \ref{shortcut}.

\begin{lemma} \label{lemma_BH} For any $q \in [0,1]$, we have $\alpha q_\alpha(B_{q}) \leq \hat\pi_\alpha q$, with equality if and only if $h_\alpha q=0$.
\end{lemma}

Remember that $\hat\pi_\alpha$ is defined in (\ref{hatpi}). Since $\hat\pi_\alpha < 1$ if $q\leq\alpha$ and $b_q>0$, Lemma \ref{lemma_BH} strictly improves upon the bound $q/\alpha$ that was obtained (when $\alpha>0$) for the tail probability of $\pi(B_q)$ by \cite{Lehmann2005} using Markov's inequality.  Through the factor $\hat\pi_\alpha$ the bound of Lemma \ref{lemma_BH} incorporates an estimate of $|T|/m$, just like adaptive FDR-controlling methods do. Moreover, \cite{Lehmann2005} did not prove that the bound is simultaneous over $S$, and consequently over $q$, a statement that follows immediately from (\ref{eq_confidence}). Additionally, from Lemma \ref{lemma_BH} we see that
\[
q_\alpha(B_{q\alpha}) \leq \hat\pi_\alpha q.
\]
Apparently, by reducing the level of BH by a factor $\alpha$ we can guarantee that the probability that the FDP exceeds the nominal level $q$ is at most $\alpha$.

\cite{Goeman2011} propose to estimate $\pi(S)$ by $\hat q(S) = q_{1/2}(S)$. This estimate is simultaneously median unbiased in the one-sided sense that it underestimates $\pi(S)$ for at least one $S$ at most 50\% of the time. By the properties of BH, the expected FDP of $B_q$ is at most $q$. By Lemma \ref{lemma_BH} we see that the estimated FDP of this set is given by
\begin{equation} \label{eq:hatq}
\hat q(B_q) = q_{1/2}(B_q) \leq 2q\hat\pi_{1/2}.
\end{equation}
This is only an upper bound, and the actual value of $\hat q(B_\alpha)$ can be much smaller, even 0 if $B=R$. Note also that $\hat\pi_{1/2} = \hat q(\{1,\ldots, m\})$ is an estimate of the proportion of true null hypotheses in the testing problem at the lenient level $\alpha=1/2$. Equation (\ref{eq:hatq}) says that if we compare the FDP estimate $\hat q(B_q)$ with the value $q$ that we might hope to see from the properties of BH, this estimate is worse by at most a factor 2. This is a surprisingly small price to pay if we realize that $\hat q (S)$ is guaranteed to be simultaneously median unbiased not only for $S=B_q$, but also for $2^m-1$ other sets.

The lemmas in this section relate the power of Simes-based closed testing to BH. Lemma \ref{lemma_Z} tempers expectations by saying that the number of discoveries in a set $S$ is always less than the number of BH-discoveries within $S$. However, Lemma \ref{lemma_BH} implies that Simes-based closed testing is at least as powerful as a BH procedure applied on a reduced $\alpha$-level. This reduction of $\alpha$ does not depend on $m$, which may suggest similar scalability for large multiple testing problems of Simes-based closed testing and BH. We will make this more precise in the next section.

\section{Scalability of power} \label{sec_scalability}

The scalability of BH has been extensively studied by \cite{Genovese2002} and \cite{Chi2007} under the assumption of independent $p$-values. They showed that if a minimal amount of signal is present, which does not depend on $m$, the power of BH per hypothesis does not vanish as $m \to\infty$. In this section we show that an analogous property holds for Simes-based closed testing.

Following \cite{Chi2007} we use the model of \cite{Efron2012} in which $p$-values are sampled independently from a mixture distribution. For $0 \leq x \leq 1$, the distribution of the $p$-values is
\begin{equation} \label{efron}
P_n(x) = \gamma x + (1-\gamma) P_{1,n}(x),
\end{equation}
where $\gamma$ is the proportion of true hypotheses, and $P_{1,n}(x) \geq x$ is the (average) distribution of the $p$-values for false null hypotheses, which depends on the sample size $n$. In the rest of this section we assume fixed $n$ and suppress the subscript. Efron's model allows $m$ to grow while ensuring that later $p$-values are `similar' to previous ones.
The assumption of independent $p$-values is common in the literature on large-scale multiple testing problems. We can, however, easily relax independence to any form of weak correlation that still guarantees uniform convergence of the empirical distribution of the $p$-values as in the theorem of Glivenko-Cantelli. Such conditions are explored e.g. by \cite{Azriel2015, Delattre2016}.

We start by studying the behavior of $\hat\pi = h/m$ as $m\to\infty$.

\begin{lemma} \label{lemma_h_asymptotic}
For $m\to\infty$, we have $\hat\pi_\alpha \toP \bar\pi_\alpha$, where
\[
\bar\pi_\alpha = \inf_{0\leq x < 1}\frac{1-P(x\alpha)}{1-x} > 0,
\]
if $P(\alpha) < 1$ and $\bar\pi_\alpha = 0$ otherwise.
\end{lemma}

The limit $\bar\pi$ is strictly smaller than 1 if and only if
\[
P(x\alpha) > x,
\]
for at least one $0 \leq x < 1$, or $P(\alpha)=1$. We call a distribution function $P$ for which $\bar\pi_\alpha < 1$ \emph{Simes-detectable at $\alpha$}. \cite{Chi2007} studied Simes-detectability under the name criticality in the context of BH. Simes-detectability holds for all $\alpha>0$ if $\gamma<1$ and test statistics are standard normal under the null hypothesis and normal with $\sigma^2=1$ and $\mu>0$ under the alternative. In other models, e.g.\ with $t$ and $F$-distributed test statistics, Simes-detectability holds from a certain sample size if $\gamma<1$. Sample sizes from which Simes-detectability is guaranteed may be calculated for these models from the results of \cite{Chi2007}. Simes-detectability never depends on $m$.

To illustrate that FWER-controlling methods generally do not scale well with $m$ we revisit Hommel's method and its rejected set $R$ defined in (\ref{def_R}), and given by
\[
R = \{1\leq i\leq m\colon hp_i \leq \alpha\}
\]
due to Lemma \ref{CTeasy}. For that set we have as $m\to\infty$, that unless we are in the unusual case that $P(\alpha) = 1$ or $P(0)>0$,
\begin{equation} \label{eq_hommel_large_m}
|R|/m \toP 0,
\end{equation}
since $hp_i \toP \infty$. This has implications for the power of Hommel's  FWER-controlling method. If $\gamma<1$, as $m$ increases, the number of false hypotheses grows proportional to $m$ in the model (\ref{efron}). The number of rejections of Hommel's method, however, vanishes relative to to $m$. Therefore, the average power, defined as the expected proportion of false null hypotheses that is rejected, vanishes. This kind of behavior is typical of FWER-controlling methods.

In contrast, BH does scale well with $m$ provided $P$ is Simes-detectable. \cite{Chi2007} showed that if $P$ is Simes-detectable at $\alpha$, we have that $|B|/m$ does not vanish. In contrast, if $P(x\alpha) < x$ for all $0\leq x<1$, we have $|B|/m \toP 0$. The borderline case is complex and we do not consider it here. Theorem \ref{thm_scalability} proves that the bounds $q_\alpha(S)$ have a similar scalability property. It is a consequence of Lemmas \ref{lemma_BH} and \ref{lemma_h_asymptotic}.

\begin{theorem} \label{thm_scalability}
Let $q \in [0,1]$ and suppose $P$ is Simes-detectable at $q\alpha$. Then there is a sequence of sets $(J_m)_{m\geq1}$ such that $J_m \subseteq \{1,\ldots,m\}$, and $q_\alpha(J_m) \leq \bar\pi_\alpha q$ for all $m$, and $\lim_{m\to\infty} \mathrm{P}(|J_m|/m \geq y) = 1$ for some $y>0$.
\end{theorem}

Theorem \ref{thm_scalability} shows that FDP confidence statements from Simes-based closed testing scale well with $m$. If Simes-detectability holds and $m$ is large enough, then for any desired FDP level $q$ we can always select a set for which we have $(1-\alpha)$-confidence of an FDP at most $q$. The size of this set does not vanish but grows linearly with $m$. Clearly, this is the analogue of Chi's result for BH: like that method, Simes-based closed testing has non-vanishing average power as $m \to\infty$ when demanding $1-\alpha$ confidence of an FDP at most $q$. The requirement for control of FDP with $1-\alpha$ confidence at level $q$ (rather than $\bar\pi_\alpha q$) is Simes-detectability at level $q\alpha/\bar\pi_\alpha$ (or $\bar\pi_\alpha=0$). This can be more strict than the requirement for FDR control of BH, which is Simes-detectability at level $q$, since usually $\alpha/\bar\pi_\alpha < 1$. In practice, this means for most models that a larger sample size is needed.

To make the link with average power explicit, choose values of $\alpha$ and $q$, and let $J_m$ be the set indicated by Theorem \ref{thm_scalability}. Consider a multiple testing procedure that rejects the $|J_m|$ hypotheses with smallest $p$-values. By Theorem \ref{thm_scalability}, such a procedure guarantees that FDP $<q$ with probability $\geq 1-\alpha$. The average power of this procedure, i.e.\ the proportion of false hypotheses that is rejected, is
\[
\frac{|J_m \cap T^c|}{|T^c|}\ =\
\mathrm{P}(i \in J_m \mid i \notin T)\ \geq\  \mathrm{P}(i \in J_m)\ =\ \frac{|J_m|}{m},
\]
where $T^c = \{1,\ldots,m\}\setminus T$. As $m\to\infty$, average power is therefore at least $c>0$ with probability 1.

Comparing Simes-based closed testing with Hommel's procedure, we see that Theorem \ref{thm_scalability} shows that there can be a huge qualitative difference in large $m$ behavior between $t_\alpha(S)$ and $|S \setminus R|$, whose relationship we discussed for small $m$ in Section \ref{sec_Non-consonant}. Theorem \ref{thm_scalability} shows that these two quantities may have divergent behavior if $m \to\infty$. By (\ref{eq_hommel_large_m}), for any $S_m$ with $|S_m|/m \to x \in (0,1]$, we have $|S_m\setminus R|/|S_m| \toP 1$, because $|R|/m \toP 0$. The result of Theorem \ref{thm_scalability} clearly does not hold if we substitute $|J_m\setminus R|$ for $t_\alpha(J_m)$. This qualitative difference in behavior as $m\to\infty$ between power of Hommel's FWER-controlling procedure and the FDP confidence bounds $t_\alpha(S)$ is exclusively due to the exploitation of non-consonant rejections by the latter. From this result, we see that non-consonant rejections, which are rare in closed testing with Simes if $m$ is small, become more and more frequent as $m$ increases, and that as $m\to\infty$ the number of consonant rejections even vanishes relatively to the number of non-consonant rejections. The step from FWER to FDP control therefore forces a radical reassessment of the value of consonance in closed testing procedures. Non-consonant procedures, which are known to be suboptimal from a FWER perspective \citep{Bittman2009, Gou2014}, seem to be crucial for powerful FDP inference in large problems.

\section{Consistency in large-scale multiple hypothesis testing}

The result in Theorem \ref{thm_scalability} is a finite sample result, which applies for $m\to\infty$ whenever the sample size is large enough for Simes-detectability to hold. However, it only indicates the existence of at least one sequence of sets $J_m$ with the desired property. FDP confidence bounds $q_\alpha(S)$, however, are available for all sets $S$. In this section we investigate consistency of $q_\alpha(S)$ as an estimator of $\pi(S)$. Consistency here is understood in the classical sense where the sample size $n\to\infty$. Since we are interested in large multiple testing problems, however, we will also let $m\to\infty$ with $n$ at an arbitrarily fast rate. In this section we emphasize dependence of quantities on $m$ and $n$

We remain within the context of Efron's model (\ref{efron}), but consider the infinite set $S \subseteq \mathbb{N}$ as the subset of interest. In Efron's model it is natural to think of $S$ as random, like $T$. We assume that conditional on $i\in S$, $p_{i,n}$ is independently drawn from a mixture model with distribution
\begin{equation} \label{efron_S}
P_n^{S}(x) = \gamma_{S} x + (1-\gamma_{S}) P_{1,n}^{S}(x).
\end{equation}
All distributions depend on $n$ (but not on $m$). The unconditional distribution of the $p$-values still follows Efron's model (\ref{efron}). We emphasize that the proportion of true null hypotheses, as well as the distribution of $p$-values under the alternative may depend on $S$. The latter may happen e.g.\ because effect sizes in $S$ may be larger or smaller than overall.

We make two further assumptions. First, the test is consistent as $n\to\infty$. Formally, for $0 < x \leq 1$,
\begin{equation} \label{ass:alternative}
\lim_{n\to\infty} P_{1,n}^{S} (x) = 1,
\end{equation}
and $P_{1,n}(x)$ is weakly increasing in $n$ for every $0 \leq x \leq 1$. We make the same assumption for $P_{1,n}(x)$. Note that we do not make the assumption that $P_{1,n}(0) \toP 1$. In most applications we expect $P_n(0)=0$ for all $n$: $p$-values are almost surely strictly positive. Second, we assume that if $S_m = S \cap \{1, \ldots, m\}$,
\begin{equation} \label{ass:largeS}
\plim_{m\to\infty} |S_m|/m = c > 0.
\end{equation}
The size of $S_m$ does not vanish relative to the size of the multiple testing problem.

Under these assumptions we have the following theorem.

\begin{theorem} \label{thm:consistency}
Under assumptions (\ref{ass:alternative}) and (\ref{ass:largeS}) we have, for $0<\alpha <1$,
\[
\plim_{(m,n)\to\infty} q_{\alpha, m, n}(S_m) = \gamma_S.
\]
\end{theorem}

Since $\gamma_S = \plim_{m\to\infty} \pi(S_m)$ in model (\ref{efron_S}) is the proportion of true null hypotheses in $S$, Theorem \ref{thm:consistency} says that $q_{\alpha,m,n}(S_m)$ is a consistent estimator of $\pi(S_m)$. This implies consistency of the confidence interval $[0, q_{\alpha,m,n}(S_m)]$ for $\pi(S_m)$, or at least the best we can expect of such consistency for a one-sided interval.

It is interesting that the theorem not only allows $m$ to go to infinity, but actually requires it. For fixed $m$ it can be shown that the probability that $q_{\alpha,m,n}(S_m)$ exceeds $\pi(S_m)$ by any positive margin vanishes as $n \to\infty$, but the same does not hold for the probability that $q_{\alpha,m,n}(S_m)$ is smaller than $\pi(S_m)$ by such a margin. Although both $m$ and $n$ need to go to infinity for consistency to hold, the relative rates at which this happens are arbitrary, as emphasized by the double limit.

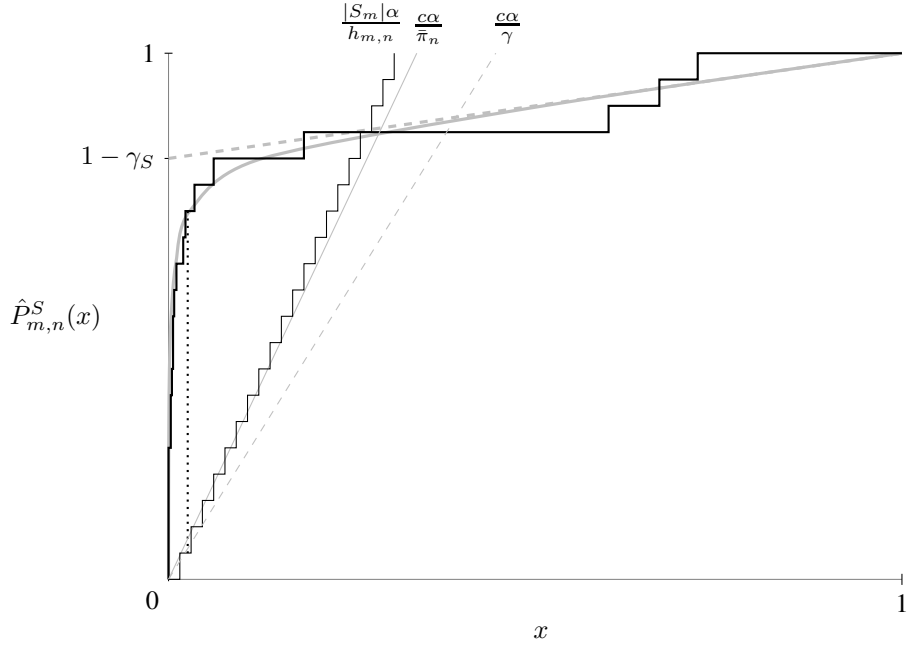
\begin{figure}[!ht]
\centering
\begin{tikzpicture}[yscale=0.7, xscale=1.5]
    \draw[help lines] (0,0)--(0,10);
    \draw[help lines] (0,0)--(6.5,0) ;
    \draw (-1pt, 10) node[left] {1} -- (1pt, 10);
    \draw (-1pt, 8) node[left] {$1-\gamma_S$} -- (1pt, 8);
    \draw [lightgray, very thick] plot [smooth] coordinates {
    (0,0) (0, 1) (.001,2) (.003,3) (.008,4) (.023,5)
    (.064, 6) (.188, 7) (.832, 8) (3.324, 9) (6.5, 10)};
    \draw [lightgray, dashed, very thick] (0,8) -- (6.5,10);
    \draw [thick] plot [const plot] coordinates {
    (0,0) (0, 2.5) (.02,3.5) (.03,4) (.04,5) (.05, 5.5) (.07, 6)
    (.13, 6.5) (.15, 7) (.23, 7.5) (.40, 8) (1.2, 8.5) (3.90, 9) (4.35, 9.5) (4.69, 10) (6.5, 10)};
    \draw [dotted, thick] (0.17,7) -- (0.17,.5);
    \draw [lightgray] (0, 0) -- (2.2, 10) node[above, black] (alpha) {$\phantom{m}\frac{c\alpha}{\bar\pi_n}$};
    \draw plot [const plot] coordinates {
    (0,0) (.1,.5) (.2,1) (.3, 1.5) (.4,2) (.5,2.5) (.6,3) (.7,3.5) (.8,4) (.9,4.5) (1,5)
    (1.1,5.5) (1.2,6) (1.3, 6.5) (1.4,7) (1.5,7.5) (1.6,8) (1.7,8.5) (1.8,9) (1.9,9.5)
    (2,10)} node[above] (alpha) {$\frac{|S_m|\alpha}{h_{m,n}}\phantom{mm}$};
    \draw [lightgray, dashed] (0, 0) -- (2.9, 10) node[above, black] (alpha) {$\phantom{m}\frac{c\alpha}{\gamma}$};
    \draw (6.5, -3pt) node[below] {1} -- (6.5, 3pt);
    \draw (0,-1pt) node[below left] {0} -- (0,0);
    \path (3.3,-1) node{$x$};
    \path (-1,5) node{$\hat P^S_{m,n} (x)$};
\end{tikzpicture}
\caption{Proof sketch for Theorem \ref{thm:consistency}. The thick step function is the empirical distribution function \mbox{$\hat P^S_{m,n}(x) = |\{i \in S_m\colon p_{i,n} \leq x\}|/|S_m|$} of the $p$-values in $S_m$. The value of $1- q_{\alpha,m,n}(S_m)$ is the maximal distance (dotted) between $\hat P^S_{m,n}(x)$ and the regular step function (left-continuous) ending at $|S_m|\alpha/h_{m,n}$. As $m\to\infty$, $\hat P^S_{m,n}(x)$ converges to the gray curve $P_n(x)$, and the regular step function converges to the gray line ending at $c\alpha/\bar\pi_n$. As $n\to\infty$, the gray curve and the gray line converge to the dashed gray lines. It is clear that $q_{\alpha,m,n}(S_m) \toP \gamma_S$ as both $m, n \to \infty$, and it is intuitive from this construction that it does not matter at what relative rate this happens.} \label{fig_cons}
\end{figure}

A sketch of the proof of Theorem \ref{thm:consistency} is given in Figure \ref{fig_cons}. From this, it is clear that the assumption in (\ref{ass:largeS}) that $c>0$ is crucial. If $c=0$, the endpoint $\alpha|S_m|/h_{m,n}$ of the reference line goes to 0 as $m\to\infty$. The rate at which this happens then competes with the rate at which $\hat P^S_{n}(x)$ goes to $1-\gamma_S$ for $x$ near 0. In that case the relative rates at which $m$ and $n$ go to infinity matter, and the double limit result fails. In particular, the assumption that $c>0$ excludes singleton sets $S$, so that Theorem \ref{thm:consistency} cannot be applied to classical FWER control. Similarly, the theorem fails if we replace $q_{\alpha,m,n}(S_m)$ by $|S_m \setminus R_{\alpha,m,n}|/|S_m|$. In practice, we expect the speed of convergence of $q_{\alpha,m,n}(S_m)$ to $\gamma_S$ to depend crucially on $c$, with large sets overestimating $\pi(S_m)$ much less than small ones if $n$ is small.

It is important that Theorem \ref{thm:consistency} holds for every fixed $0<\alpha\leq 1$. If consistent estimation, rather than confidence, is the primary concern, larger values of $\alpha$, such as $\alpha=1/2$, can be expected to yield faster convergence. Note the contrast with comparable results, e.g.\ for BH, where for fixed $\alpha$ it can easily be shown under the model (\ref{efron}) \citep{Chi2007} that, if $\gamma<1$,
\[
\plim_{(m,n)\to\infty} \frac{|B_{\alpha,m,n}|}{m} = \frac{1-\gamma}{1-\alpha\gamma},
\]
implying that, if $\gamma > 0$, $B_{\alpha,m,n}$ is always asymptotically `too large' for fixed $\alpha$. For $B_{\alpha,m,n}$ to consistently approach the set of false null hypotheses, we need $\alpha$ to approach 0, but at a slow enough rate in $n$ not to upset the asymptotic behavior of $|B_{\alpha,m,n}|$. Such subtlety is not necessary with Theorem \ref{thm:consistency}, which allows fixed $\alpha$.

\section{Discussion}

Closed testing can be used for FWER control of elementary hypotheses, but also for creating simultaneous confidence intervals for the proportion of true null hypotheses in subsets. This use of closed testing implies FWER control over exponentially many hypotheses of the form
\[
\pi(S) \geq q
\]
for all $S \subseteq \{1,\ldots, m\}$ and all $q \in [0,1]$. One might expect that such a huge multiple testing burden would be crippling in large problems, and that power of such a procedure would vanish for large $m$, as indeed it does when using closed testing for FWER control on elementary hypotheses.

However, in this paper we have shown that this is not generally the case. For closed testing with Simes local tests the power for inferring on other hypotheses than elementary ones does not vanish as $m$ increases. We have presented two concrete results under the assumption of independent $p$-values. First, we have proven that at any desired FDP level and any desired confidence level, as $m \to \infty$, it is possible to find a set that has confidence of an FDP below the desired level. The size of this set grows to infinity with $m$. This result holds whenever the sample size $n$ is large enough to have Simes-detectability. Second, we show that for every set whose size does not vanish relative to $m$, the confidence bound for FDP converges to the true FDP whenever both $m$ and $n$ go to infinity, at arbitrary relative rates. If the size of the set does not vanish, the power for statements on FDP of such sets does not vanish.

For fixed $m$ and $n$ we have established relationships between the procedure of Benjamini and Hochberg (BH) and closed testing with Simes local tests that go beyond the trivial observation that the methods have the same weak FWER control. The performance of FDP confidence bounds can be bounded from above by the BH-rejected set, and from below by the BH-rejected set at an $\alpha$-level reduced by a fixed factor. The loss of power of FDP confidence statements relative to FDR-control is surprisingly small considering that many more and much stronger statements are made by closed testing.

Our results force a reassessment of the value of consonance in closed testing. Crucial for the results of this paper are the non-consonant rejections that the Simes test makes. Non-consonant rejections are rare for small $m$, but quickly start to dominate as $m$ increases. Consonant methods will not generally have the scalability in $m$ that we have shown here. For example, it is easy to see that the asymptotic results of this paper fail for closed testing with Bonferroni local tests.

Closed testing can be computationally expensive. There is no value in having good power for large $m$ if computation time is prohibitive. We have presented a fast and exact shortcut that allows calculation of the confidence bound $q_\alpha(S)$ in linear time in $|S|$ after an initial preparatory step (independent of $S$) that takes $O(m\log m)$ time. This shortcut improves upon earlier ones in computation time, applicability and exactness. It is implemented in the package \texttt{hommel} in \texttt{R} \citep{Goeman2017}.

\appendix

\section{Proofs}

\begin{replem}{\ref{CTeasy}}
$I \in \mathcal{X}$ if and only if there exists some $1 \leq i \leq |I|$ such that
$h p_{(i\mathbin{:}I)} \leq i\alpha$.
\end{replem}

\begin{proof}
Consider first the case $|I| > h$. By Lemma \ref{h}, $I \in \mathcal{X} \subseteq \mathcal{U}$, so there is an $1 \leq i \leq |I|$ such that
$h p_{(i\mathbin{:}I)} \leq |I| p_{(i\mathbin{:}I)} \leq i \alpha$.

Next, consider the case $|I| \leq h$. First, suppose that there exists some $1 \leq i \leq |I|$ with $h p_{(i\mathbin{:}I)} \leq i\alpha$. Take any $J \supseteq I$. If $|J| > h$, we have $J \in \mathcal{U}$ by Lemma \ref{h}. If $|J| \leq h$, we have
\[
|J|p_{(i\mathbin{:}J)} \leq hp_{(i\mathbin{:}J)} \leq hp_{(i\mathbin{:}I)} \leq i \alpha,
\]
which proves that $J \in \mathcal{U}$. Since this holds for all $J \supseteq I$, we have $I \in \mathcal{X}$.

Next, suppose that there is no $1 \leq i \leq |I|$ with $h p_{(i\mathbin{:}I)} \leq i \alpha$. Let $J = I \cup K_j$ for some $0\leq j \leq h$ such that $|J| = h$. If $1 \leq i \leq h-j$, by the assumption we have
\[
|J|p_{(i\mathbin{:}J)} = hp_{(i\mathbin{:}J)} = hp_{(i\mathbin{:}I)} > i \alpha.
\]
If $h-j \leq i \leq h$, since $K_j \subseteq K_h$ we have
\[
|J|p_{(i\mathbin{:}J)} = |K_h|p_{(i\mathbin{:}J)} = |K_h|p_{(i\mathbin{:}K_h)} > i\alpha
\]
because $K_h \notin \mathcal{U}$. Taken together, this implies that $J \notin \mathcal{U}$, so $I \notin \mathcal{X}$.
\end{proof}

\begin{repthm}{\ref{shortcut}}
If $S \neq\emptyset$ we have $d_\alpha(S) = \max_{1\leq u \leq |S|} 1-u+|\{i\in S\colon hp_i \leq u\alpha\}|$.
\end{repthm}

\begin{proof}
We will use
\begin{eqnarray*}
d_\alpha(S) &=& |S| - t_\alpha(S) \\
&=& \min\{|S\setminus I|\colon I \subseteq S, I \notin \mathcal{X}\} \\
&=& \min\{|I|\colon I \subseteq S, S\setminus I \notin \mathcal{X}\}.
\end{eqnarray*}
By Lemma \ref{CTeasy} we have $S \setminus I \in \mathcal{X}$ if and only if for some $1 \leq u \leq |S\setminus I|$ we have $|\{i \in S\setminus I\colon hp_i \leq u\alpha\}| \geq u$, and we may trivially extend the range to $1 \leq u \leq |S|$. Thus, $S \setminus I \notin \mathcal{X}$ if and only if for all such $u$ we have $|\{i \in S\setminus I\colon hp_i \leq u\alpha\}| \leq u-1$. That is, for all such $u$,
\begin{equation} \label{eq_proof_thm1}
|\{i \in S\colon hp_i \leq u\alpha\}| - u + 1 \leq |\{i \in I\colon hp_i \leq u\alpha\}|.
\end{equation}
Denote the left-hand side of (\ref{eq_proof_thm1}) by $g(u)$ and the right-hand side by $f(I, u)$. Let $d = \max_{1 \leq u \leq |S|} g(u)$. Since $0 \leq d \leq |S|$ we can pick $I = S \cap L_j$ for some $0 \leq j \leq m$ such that $|I| = d$. For all $1 \leq u \leq |S|$, we have $f(I,u) \leq |I| = d$. If $f(I,u) < d = |I|$, then $g(u) \leq f(S,u) = f(I,u)$, where the latter step follows by construction of $I$. If $f(I,u) = d$, then $g(u) \leq d = f(I,u)$. We conclude that $I$ satisfies (\ref{eq_proof_thm1}) for all $1 \leq u \leq |S|$. Obviously, (\ref{eq_proof_thm1}) cannot hold for any $I$ with $|I| < d$. We conclude that $d_\alpha(S) = d$.
\end{proof}

\begin{replem}{\ref{lemma_z}}
The maximum of $1 - u + |\{i\in S\colon hp_i \leq u\alpha\}|$ over $1 \leq u \leq |S|$
is attained for some $1\leq u \leq z-m+h+1$, where
\[
z_\alpha = \left\{ \begin{array}{ll}
0, & \textrm{if $h=m$;} \\
\min\big\{ m-h \leq i \leq m \colon hp_{(i)} \leq (i-m+h+1)\alpha \big\}, & \textrm{otherwise.}
\end{array} \right.
\]
\end{replem}

\begin{proof}
Let us first verify that $z$ is well-defined. We write out the definition of $h$ from (\ref{def_h}):
\begin{equation}\label{hommel}
h =\max \big\{i\in\{0,\ldots,m\}\colon ip_{(m-i+j)}>j\alpha, \textrm{ for } j=1,\ldots,i\big\}.
\end{equation}
Since $h$ is a maximum, we have for $h<m$ that there is at least one $1\leq i \leq h+1$
such that $hp_{(m-h+i-1)} \leq (h+1)p_{(m-h+i-1)} \leq i\alpha$, so we have $z=m-h+i-1$ for the smallest such $i$, so $z$ is well-defined.

Next, pick any $z-m+h+1 \leq u \leq m$. To prove the lemma we will show that $g(z-m+h+1) \geq g(u)$, where $g(u)$ is the function we are maximizing. By definition of $h$, we have $hp_{(m-h+i)} > i\alpha$ for all $1\leq i \leq h$, which becomes, after changing variables, $hp_{(i)} > (i-m+h)\alpha$ for all $m-h+1 \leq i \leq m$. In particular, if $h=m$, then $g(u) \leq 0 \leq g(1)$ because $|\{i \in S\colon hp_i \leq u\alpha\}| \leq u-1$, so assume $h<m$. Consider any $1 \leq i \leq m$. If $i \leq z$, then $hp_{(i)} \leq hp_{(z)} \leq (z-m+h+1)\alpha$ by definition of $z$. If $i \geq u+m-h > z$, then $i \geq m-h+1$, so $hp_{(i)} > (i-m+h)\alpha \geq u\alpha$ by definition of $z$. Hence,
\[
|\{i \in S\colon (z-m+h+1)\alpha < hp_i \leq u\alpha\}| \leq u+m-h-z-1,
\]
from which the conclusion follows.
\end{proof}

\begin{replem}{\ref{lemma_Z}}
We have
\begin{enumerate}
\item $d_\alpha(S)=d_\alpha(S \cap Z)$ for all $S$;
\item $Z$ is minimal, i.e.\ if $d_\alpha(S)=d_\alpha(S \cap Y)$ for all $S$, then $Z \subseteq Y$;
\item $d_\alpha(Z)=m-h$;
\item $Z\subseteq B$.
\end{enumerate}
\end{replem}

\begin{proof}
The entire lemma is trivial if $h=m$. Assume $h<m$. In that case we remark that by definition of $h$ we have $h p_{(i)} > (i-m+h)\alpha$ for all $m-h+1 \leq i \leq m$. Thus, for any $i \in S \setminus Z$ we have $hp_i \geq hp_{(z+1)} > (z-m+h+1)\alpha$ since $z \geq m-h$. Now Statement 1 follows directly from Lemma \ref{lemma_z}.

By Statement 1, $d_\alpha(Z) = d_\alpha(Z \cap \{1,\ldots,m\}) = d_\alpha(\{1,\ldots,m\}) = m-h$, where the final equality holds by definition of $d_\alpha$ and Lemma \ref{h}. This proves Statement 3.

Now take any $Y$ that has the property of Statement 2. Consider $S = Y \cap Z$. We have
\begin{equation} \label{eq_proof_z}
m-h = d_\alpha(Z) = d_\alpha(S) = \max_{1 \leq u \leq |S|} \big\{1-u+|\{i \in S\colon hp_i \leq u\alpha\}| \big\}.
\end{equation}
By the definition of $z$ as a minimum we have that for all $m-h \leq j < z$ that \mbox{$h p_{(j)} > (j-m+h+1)\alpha$}. Therefore, $|\{i \in S\colon hp_i \leq (u-m+h+1)\alpha\}| < u$, for all $m-h \leq u < z$. Changing variables, we obtain $|\{i \in S\colon hp_i \leq u\alpha\}| < u+m-h-1$ for all $1 \leq u < z-m+h+1$, so that $1-u+|\{i \in S\colon hp_i \leq u\alpha\}| < m-h$ for all $1 \leq u < z-m+h+1$. Therefore, and by Lemma \ref{lemma_z}, the maximum in (\ref{eq_proof_z}) is attained at $u=z-m+h+1$, and we obtain from (\ref{eq_proof_z}) that
\[
|\{i \in S\colon hp_i \leq (z-m+h+1)\alpha\}| = z,
\]
so that $|S| = |Z|$ and $Z\subseteq Y$. This proves Statement 2.

To prove Statement 4, note that $b=z=0$ if $h=m$. If $h < m$, by definition of $h$ there is a $1\leq j \leq h+1$ such that $(h+1)p_{(j+m-h-1)} \leq j\alpha$. Let $z' = j+m-h-1$. Then
\[
hp_{(z')} \leq (h+1)p_{(z')} \leq (z' -m+h+1)\alpha,
\]
so $z\leq z'$ by definition of $z$. Now,
\begin{eqnarray*}
mp_{(z')} &=& \frac{m(h+1)p_{(z')}}{(h+1)} \leq \frac{m(z' -m+h+1)}{(h+1)z'}z'\alpha \\ &=& \frac{mz' - m(m-h-1)}{mz' - z'(m-h-1)}z'\alpha \leq z'\alpha,
\end{eqnarray*}
so $z \leq z' \leq b$ by definition of $b$.
\end{proof}

\begin{replem}{\ref{lemma_BH}}
For any $q \in [0,1]$, we have $\alpha q_\alpha(B_{q}) \leq \hat\pi_\alpha q$, with equality if and only if $\hat\pi_\alpha q=0$.
\end{replem}

\begin{proof}
If $b_q = 0$ or $\alpha=0$, then $\alpha q_\alpha(B_q) = 0$ by definition.
If $\hat\pi_\alpha=0$, then $\mathcal{X} = \mathcal{I}$, so $q_\alpha(B_q) = 0$.
If $q=0$, then $p_i=0$ for all $i \in B_q$, so $\mathcal{X}$ contains every nonempty subset of $B_q$, so $q_\alpha(B_q) = 0$.
Otherwise, pick any $J \subseteq B_q$ with $|J| \geq b_q \hat\pi_\alpha q/\alpha > 0$. By definition of $B_q$ we have $p_{(b_q)} \leq b_q q/m$, so
\[
hp_{(|J|\mathbin{:}J)} \leq hp_{(b_q)} \leq b_q \hat\pi_\alpha q \leq |J|\alpha.
\]
Hence, $J \in \mathcal{X}$ by Lemma \ref{CTeasy}, and $q_\alpha(B_q) \leq \hat\pi_\alpha q/\alpha - 1/b_q$.
\end{proof}

\begin{replem}{\ref{lemma_h_asymptotic}}
For $m\to\infty$, we have $\hat\pi_\alpha \toP \bar\pi_\alpha$, where
\[
\bar\pi_\alpha = \inf_{0\leq x < 1}\frac{1-P(x\alpha)}{1-x} > 0,
\]
if $P(\alpha) < 1$ and $\bar\pi_\alpha = 0$ otherwise.
\end{replem}

\begin{proof}
If $P(\alpha) = 1$, we have $\mathrm{P}(h=0)=1$ and the second statement follows immediately.

Assume $P(\alpha) < 1$. To prove the first statement we will bound $h$ from above and below before taking the limit. Note that if $i>0$ and $1\leq j \leq i$, we have that the condition $ip_{(m-i+j)} > j\alpha$ from (\ref{hommel}) is equivalent to $|\{1\leq k \leq m\colon ip_k > j\alpha\}| > i-j$.

If $h > 0$, take any $x \in [0,1)$ and let $j = \lfloor hx \rfloor +1$. Then by definition of $h$ we have
\begin{eqnarray*}
|\{1\leq k \leq m\colon p_k > x\alpha\}| &\geq& |\{1\leq k \leq m\colon hp_k > j\alpha\}| \\
&\geq& h-j+1\\
&\geq& h(1-x).
\end{eqnarray*}
Hence,
\[
\inf_{0\leq x<1} \frac{|\{1\leq k \leq m\colon p_k > x\alpha\}|}{1-x} \geq h.
\]

If $0 < h < m$, by definition of $h$ there is some $1 \leq j \leq h+1$ such that
\begin{equation} \label{hmj}
|\{1\leq k \leq m\colon (h+1)p_k > j\alpha\}| \leq h+1-j.
\end{equation}
If $j=h+1$ we have $p_{(m)} \leq \alpha$ and consequently $h=0$ by (\ref{hmj}), so we can assume
$j \leq h$. Take $x=j/(h+1)$. Then $0 < x < 1$ and
\begin{eqnarray*}
|\{1\leq k \leq m\colon p_k > x\alpha\}| &=& |\{1\leq k \leq m\colon (h+1)p_k > j\alpha\}| \\
&\leq& h-j+1\\
&=& (h+1)(1-x).
\end{eqnarray*}
Consequently,
\[
\inf_{0 \leq x<1} \frac{|\{1\leq k \leq m\colon p_k > x\alpha\}|}{1-x} \leq h+1.
\]
If $h=m$, we have $p_{(1)} > 0$, so taking $x=0$ we also have
\[
\inf_{0 \leq x<1} \frac{|\{1\leq k \leq m\colon p_k > x\alpha\}|}{1-x} \leq |\{1\leq k \leq m\colon p_k > 0\}| = m < h+1.
\]

Therefore, since $h=0$ is equivalent to $p_{(m)} \leq \alpha$, we have
\[
(\phi(\hat P_m)- m^{-1})\mathds{1}_{\{p_{(m)} > \alpha\}} \leq \frac hm \leq
\phi(\hat P_m)\mathds{1}_{\{p_{(m)} > \alpha\}},
\]
where $\hat P_m (x) = m^{-1}|\{1\leq k \leq m\colon p_k \leq x\}|$ and \[\phi(f) = \inf_{0\leq x<1} \frac{1-f(x\alpha)}{1-x}.\]

By Lemma \ref{lemma_cont}, below, $\phi$ is continuous in $P$. By the Continuous Mapping Theorem $\phi(\hat P_m)$ therefore converges to $\phi(P)$ since $\hat P_m \toP P$ uniformly by Glivenko-Cantelli. We note that
\[
\mathrm{P}(p_{(m)} \leq \alpha) = P(\alpha)^m \to 0
\]
since $P(\alpha) < 1$, so also $h/m$ converges to $\phi(P)$ by Slutsky's Theorem.
\end{proof}

\begin{lemma} \label{lemma_cont}
Let $F$ be the space of weakly increasing functions from $[0,1]$ to $[0,1]$ under the uniform norm. The functional $\phi\colon F \to [0,1]$ given by
\[
\phi(f) = \inf_{0\leq x<1} \frac{1-f(x\alpha)}{1-x}
\]
is continuous in any $f$ with $f(\alpha)<1$.
\end{lemma}

\begin{proof}
Note that indeed $0 \leq \phi(f) \leq 1-f(0) \leq 1$ for every $f \in F$. We show that $\phi$ is continuous in $f$. Let $c = 1-f(\alpha) > 0$. Given $0<\eps\leq 1$, take $\delta = c\eps/2$ and pick any $g$ with $\|g-f\|_\infty < \delta$. If $x \geq 1-c/2$ we have
\[
\frac{1-g(x\alpha)}{1-x} \geq \frac{1-g(\alpha)}{c/2} > \frac{1-f(\alpha)-\delta}{c/2} = \frac{c-\delta}{c/2} = 2 - \eps \geq 1 \geq \phi(g),
\]
so
\[
\phi(g) = \inf_{0\leq x<1-c/2} \frac{1-g(x\alpha)}{1-x}.
\]
In particular the same holds for $g=f$ as well. Since for $0 \leq x < 1-c/2$ we have
\[
\Big|\frac{g(x\alpha)- f(x\alpha)}{1-x}\Big| < \frac{\delta}{c/2} = \eps,
\]
we have $|\phi(g) - \phi(f)| < \eps$ as desired.
\end{proof}

\begin{repthm}{\ref{thm_scalability}}
Let $q \in [0,1]$ and suppose $P$ is Simes-detectable at $q\alpha$. Then there is a sequence of sets $(J_m)_{m\geq1}$ such that $J_m \subseteq \{1,\ldots,m\}$, and $q_\alpha(J_m) \leq \bar\pi_\alpha q$ for all $m$, and $\lim_{m\to\infty} \mathrm{P}(|J_m|/m \geq y) = 1$ for some $y>0$.
\end{repthm}

\begin{proof}
Throughout the proof we make dependence of quantities on $m$ explicit.

Take $J_m = B_{\tilde\alpha_m, m}$, where $\tilde\alpha_m = 1$ if $\hat\pi_{\alpha,m}=0$ and $\tilde\alpha_m = q\alpha\bar\pi_\alpha /\hat\pi_{\alpha,m}$ otherwise. By Lemma \ref{lemma_BH}, and noting that $q_{\alpha,m}(B_0)=0$ for all $\alpha$, we have $q_{\alpha,m}(J_m) \leq \bar\pi_\alpha q$.

We have
\begin{eqnarray*}
|J_m|\ =\ b_{\tilde\alpha,m} &=& \max\{1\leq i \leq m\colon mp_{(i\mathbin{:}\{1,\ldots,m\})} \leq i\tilde\alpha_m\} \\
&=& \max\{1\leq i\leq m\colon |\{1 \leq k \leq m\colon p_k \leq i\tilde\alpha_m/m\}| \geq i\},
\end{eqnarray*}
so that
\begin{eqnarray*}
\frac{|J_m|}{m} &=& \max\{x \in \{1/m, \ldots, m\}\colon \hat P_m(x\tilde\alpha_m) \geq x\} \\
&\geq& \max\{0 \leq x \leq 1\colon \hat P_m(x\tilde\alpha_m) \geq x + 1/m\}.
\end{eqnarray*}

If $\bar\pi_{\alpha,m} = 0$ we have $P(\alpha) = 1$ by Lemma \ref{lemma_h_asymptotic}, so $\mathrm{P}(\hat\pi_{\alpha,m}=0) =1$ and therefore $|J_m|/m \toP 1$ as $m\to\infty$.

Otherwise, $\bar\pi_{\alpha,m} > 0$ and there is some $0 \leq x < 1$ such that $P(xq\alpha) > x$ by Simes-detectability and Lemma \ref{lemma_h_asymptotic}. Fix any $x < y < z < P(xq\alpha)$ and \mbox{$0 < \eps < P(xq\alpha)-z$}. Let $E_m$ be the event that $x\hat\pi_{\alpha,m} \leq y\bar\pi_\alpha$, and $F_m$ the event that $\|\hat P_m - P\|_\infty < P(xq\alpha) - z - \eps$. Since $E_m$ implies that $xq\alpha \leq y\tilde\alpha_m$, we have under $E_m$ and $F_m$ that
\[
\hat P_m(y\tilde\alpha_m) > P(y\tilde\alpha_m) - P(xq\alpha) + z + \eps \geq z + \eps > y + \eps.
\]
By Lemma \ref{lemma_h_asymptotic} and Glivenko-Cantelli, we get that
\[
\mathrm{P}(E_m \cap F_m) \geq \mathrm{P}(E_m)+\mathrm{P}(F_m)-1\to 1
\]
as $m\to\infty$. Hence, for all big enough $m \geq 1/\eps$ we have $\mathrm{P}(E_m \cap F_m)>0$, and we find for $m\to\infty$ that
\begin{eqnarray*}
\mathrm{P}(|J_m|/m \geq y) &\geq& \mathrm{P}(\hat P_m(y\tilde\alpha_m) \geq y + \eps) \\
&\geq& \mathrm{P}(\hat P_m(y\tilde\alpha_m) \geq y + \eps\mid E_m \cap F_m)\mathrm{P}(E_m \cap F_m) \\
&=& \mathrm{P}(E_m \cap F_m) \to 1
\end{eqnarray*}
as desired.
\end{proof}

For the next proof we make use of a result due to \cite{Hildebrandt1958}, which we restate here in a variant for convergence in probability. The proof is a simple reworking of the proof of \citeauthor{Hildebrandt1958} and is omitted.

\begin{lemma}[Hildebrandt] \label{hildebrandt}
Let $X_{m,n}$ be random variables weakly stochastically decreasing in $n$ for every $m$. If
\[
\plim_{n\to\infty} \plim_{m\to\infty} X_{m,n} = \plim_{m\to\infty} \plim_{n\to\infty} X_{m,n} = c
\]
then
\[
\plim_{(m,n)\to\infty} X_{m,n} = c.
\]
\end{lemma}

\begin{repthm}{\ref{thm:consistency}}
Under assumptions (\ref{ass:alternative}) and (\ref{ass:largeS}) we have, for $0<\alpha<1$,
\[
\plim_{(m,n)\to\infty} q_{\alpha, m,n}(S_m) = \gamma_S.
\]
\end{repthm}

\begin{proof}
We note that $q_{\alpha, m, n}(S_m)$ is weakly increasing in all $p_{i,n}$, and that all $p_{i,n}$ are weakly stochastically decreasing in $n$ by the assumption that $P_n^S(X)$ is weakly increasing in $n$, so $q_{\alpha, m,n}(S_m)$ is weakly stochastically decreasing in $n$. To prove the theorem, by Lemma \ref{hildebrandt} it suffices to show that the iterated limits are both equal to $\gamma_S$. We may assume $|S_m| > 0$.

We find upper and lower bounds for $q_{\alpha,m,n}(S_m)$. Assume $h_{m,n} > 0$. We can rewrite the statement of Theorem \ref{shortcut} as
\begin{equation} \label{rewrite d}
d_{\alpha,m,n}(S_m) = \max_{u \in (0, \infty)} \big( |\{i \in S_m\colon h_{m,n}p_{i,n} \leq u\alpha\}| - \lceil u \rceil + 1 \big).
\end{equation}
To see this, let $g(u) = |\{i \in S_m\colon h_{m,n}p_{i,n} \leq u\alpha\}| - \lceil u \rceil + 1$ and note that $g(\lceil u \rceil) \geq g(u)$ so that $g(u)$ is always maximized at an integer value of $u$. At integer values the new expression yields the same maximum as the expression in Theorem \ref{shortcut}. Note that we may freely change the range of $u$ from $(0, |S_m|)$ to $(0,\infty)$ since $g(u) \leq g(|S_m|)$ for $u > |S_m|$. Writing $u= xh_{m,n}/\alpha$, replacing maximum by supremum, we see that $q_{\alpha,m,n}(S_m)$ lies between
\begin{equation} \label{Qmn}
Q_{m,n} = 1-\sup_{x \in (0, \infty)} \Big(\hat P^S_{m,n}(x) - \frac{xh_{m,n}}{\alpha|S_m|}\Big)
\end{equation}
and $Q_{m,n} + 1/|S_m|$, where $\hat P^S_{m,n}(x) = |\{i \in S_m\colon p_{i,n} \leq x\}|/|S_m|$. Note that the same bounds on $q_{\alpha,m,n}(S_m)$ also hold if $h_{m,n}=0$. The maximum in (\ref{Qmn}) is attained for $x \leq 1$ since $P^S_{m,n}(x)=1$ for $x\geq 1$. Since $\hat P^S_{m,n}(x)$ is weakly increasing in $x$, we can therefore take $x \in [0,1]$ instead.

We start by investigating the case where $m\to\infty$ first. By Slutsky's Theorem and Assumption (\ref{ass:largeS}), the limit of $q_{\alpha,m,n}(S_m)$ is identical to that of $Q_{m,n}$. We have that $\hat P^S_{m,n}(x) \toP P^S_{n}(x)$ uniformly in $x$ by Glivenko-Cantelli, and $h_{m,n}/|S_m| \toP \bar\pi_n/c$ by Slutsky's Theorem, Lemma \ref{lemma_h_asymptotic} and Assumption (\ref{ass:largeS}). Consequently, also $\hat P^S_{m,n}(x) - xh_{m,n}/(\alpha|S_m|)$ converges uniformly over $x \in (0,1]$. Since the supremum is a continuous function, we may exchange supremum and limit by the Continuous Mapping Theorem. We obtain
\begin{equation} \label{mlimit}
\plim_{m\to\infty} q_{\alpha,n}(S_m) = 1-\sup_{x \in [0, 1]} \Big(P^S_{n}(x) - \frac{x\bar\pi_{n}}{c\alpha}\Big).
\end{equation}

To let $n\to\infty$ next, we start by investigating $\lim_{n\to\infty} \bar\pi_{n}$. Let $P(x) = 1-\gamma + x\gamma$, so that $P_n(x) \to P(x)$ (pointwise) for all $0<x\leq 1$, and $P_n(x) \leq P(x)$ for all $n$ and all $0\leq x\leq 1$. If $\gamma > 0$ we have $P_n(\alpha) \leq P(\alpha) < \alpha$, so that by Lemma \ref{lemma_h_asymptotic}, we have the lower bound
\[
\bar\pi_n \ =\ \inf_{0\leq x < 1} \frac{1-P_n(x\alpha)}{1-x}\ \geq\
\inf_{0\leq x < 1} \frac{1-P(x\alpha)}{1-x}\ =\ \gamma,
\]
and the same lower bound also holds if $\gamma=0$. For an upper bound, choose $0<y<1$ arbitrary. Then we have, as $n\to\infty$, by Lemma \ref{lemma_h_asymptotic},
\[
\gamma \ \leq\ \bar\pi_n\ \leq\ \inf_{0\leq x < 1} \frac{1-P_n(x\alpha)}{1-x}\ \leq\ \frac{1-P_n(y\alpha)}{1-y}\ \to\ \frac{1-P(y\alpha)}{1-y}\ = \ \frac{\gamma - \gamma y\alpha}{1-y}.
\]
Since $y$ was arbitrary it follows that $\lim_{n\to\infty} \bar\pi_{n} = \gamma$.

Next, we find upper and lower bounds for (\ref{mlimit}). Let $P^S(x) = 1-\gamma_S + x\gamma_S$, so that $P_n^S(x) \to P^S(x)$ (pointwise) for all $0<x\leq 1$, and $P_n^S(x) \leq P^S(x)$ for all $n$ and all $0\leq x\leq 1$. We have, as $n\to\infty$,
\begin{eqnarray*}
1-\sup_{x \in [0, 1]} \Big(P^S_{n}(x) - \frac{x\bar\pi_{n}}{c\alpha}\Big)
&\geq&
1-\sup_{x \in [0, 1]} \Big(P^S(x) - \frac{x\bar\pi_{n}}{c\alpha}\Big) \\
&=&
\gamma_S + \inf_{x \in [0, 1]} \Big( x \frac{\bar\pi_{n}}{c\alpha}  - x \gamma_S \Big) \\
&=& \gamma_S + \min\Big(0, \frac{\bar\pi_{n}}{c\alpha}  - \gamma_S\Big)
\ \to\ \gamma_S,
\end{eqnarray*}
where the final step follows because $\lim_{n\to\infty} \bar\pi_{n} = \gamma$, and $\gamma/c\alpha \geq \gamma_S$, since
\[
c\gamma_S = \plim_{m\to\infty}\frac{|T \cap S_m|}{m} \leq \plim_{m\to\infty}\frac{|T \cap \{1,\dots,m\}|}{m} = \gamma.
\]
Choose $0<y<1$ arbitrary. We have, as $n\to\infty$,
\[
1-\sup_{x \in [0, 1]} \Big(P^S_{n}(x) - \frac{x\bar\pi_{n}}{c\alpha}\Big)
\ \leq\ 1- \Big(P^S_{n}(y) - \frac{y\bar\pi_{n}}{c\alpha}\Big) \\
\ \to\ \gamma_S + y \Big(\frac{\gamma}{c\alpha} - \gamma_S\Big).
\]
Since $\gamma/c\alpha \geq \gamma_S$, and $y$ was arbitrary, we combine the upper and lower bounds with (\ref{mlimit}) to obtain
\[
\plim_{n\to\infty} \plim_{m\to\infty} q_{\alpha,m,n}(S_m) = \gamma_S.
\]

In the reverse order, letting $n \to \infty$ first, the distribution of the $p$-values converges to $p_{i,\infty} \sim P$ marginally and $p_{i,\infty} \sim P^S$ given $i \in S$.
We now proceed to follow the same arguments as for $m\to\infty$ above. By (\ref{Qmn}) we have that $\lim_{n\to\infty} q_{\alpha, m,n}(S_m)$ lies between
\[
Q_{m,\infty} = 1-\sup_{x \in (0, 1]} \Big(\hat P^S_{m,\infty}(x) - \frac{xh_{m,\infty}}{\alpha|S_m|}\Big)
\]
and $Q_{m,\infty} + 1/|S_m|$, where the distributions of $\hat P^S_{m,\infty}(x)$ and $h_{m,\infty}$ are governed by $P^S$ and $P$, respectively. By Assumption (\ref{ass:largeS}) and Slutsky's Theorem we find that
\[
\lim_{m\to\infty} \lim_{n\to\infty} q_{\alpha, m,n}(S_m) = \lim_{m\to\infty} Q_{m,\infty}.
\]
We have that $\hat P^S_{m,\infty}(x) \toP P^S(x)$ uniformly in $x$ by Glivenko-Cantelli, and $h_{m,\infty}/|S_m| \toP \gamma/c$ by Slutsky's Theorem, Lemma \ref{lemma_h_asymptotic} and Assumption (\ref{ass:largeS}). Consequently, also $\hat P^S_{m,\infty}(x) - xh_{m,\infty}/(\alpha|S_m|)$ converges uniformly over $x \in (0,1]$. Since the supremum is a continuous function, we may exchange supremum and limit by the Continuous Mapping Theorem. We obtain
\[
\plim_{m\to\infty} \plim_{n\to\infty} q_{\alpha,m,n}(S_m) = \gamma_S.
\]
The result of the theorem now follows from Lemma \ref{hildebrandt}.
\end{proof}

\bibliographystyle{chicago}
\bibliography{hommel}

\end{document}